\documentclass[aps,pra,onecolumn,
epsfig,showpacs,superscriptaddress,footnote,nofootinbib]{revtex4}

\usepackage{amsmath,amsthm,amssymb,mathrsfs,amsfonts}
\usepackage{graphicx}
\usepackage{color}
\usepackage{multirow}
\usepackage{mathtools}
\makeatletter
\def\be{\begin{equation}}
\def\ee{\end{equation}}
\def\bea{\begin{eqnarray}}
\def\eea{\end{eqnarray}}
\def\ben{\begin{equation*}}
\def\een{\end{equation*}}
\def\bean{\begin{eqnarray*}}
\def\eean{\end{eqnarray*}}
\def\bma{\begin{mathletters}}
\def\ema{\end{mathletters}}
\def\bi{\begin{itemize}}
\def\ei{\end{itemize}}
\newtheorem{thm}{Theorem}

\newtheorem{cor}[thm]{Corollary}

\newcommand{\ket}[1]{ | \, #1 \rangle}

\newcommand\Mycomb[2][^n]{\prescript{#1\mkern-0.5mu}{}C_{#2}}
\begin{document}

\title{Analytical construction of non local operator for n-qubit Dicke state}
\author{Sutapa Saha}
\email{sutapa.gate@gmail.com}\affiliation{Physics \& Applied Mathematics Unit, Indian Statistical Institute, 203 B.T. Road, Kolkata- 700108, India}
\author{Some Sankar Bhattacharya}
\email{somesankar@gmail.com}\affiliation{Physics \& Applied Mathematics Unit, Indian Statistical Institute, 203 B.T. Road, Kolkata- 700108, India}
\author{Arup Roy}
\email{arup145.roy@gmail.com}\affiliation{Physics \& Applied Mathematics Unit, Indian Statistical Institute, 203 B.T. Road, Kolkata- 700108, India}
\author{Amit Mukherjee}
\email{amitisiphys@gmail.com}\affiliation{Physics \& Applied Mathematics Unit, Indian Statistical Institute, 203 B.T. Road, Kolkata- 700108, India}
\author{Ramij Rahaman}
\email{ramijrahaman@gmail.com}\affiliation{Department of Mathematics, University of Allahabad, Allahabad 211002, U.P., India}


\begin{abstract}
Entanglement in multipartite quantum systems is much more elusive than its bipartite counterpart. In recent past the usefulness of multipartite entangled states in several information theoretic tasks have been demonstrated. Being a resource, the detection of multipartite entanglement is an imperative necessity. Among the different classes of multipartite entangled states the Dicke state has found importance in several tasks due to its permutation symmetric nature. In this work we propose a simple and elegant way of detecting n qubit Dicke states using permutation symmetric Bell operators. We conjecture that maximal expectation value of the operator corresponds to the detection of Dicke states.
\end{abstract}

\maketitle
\section{Introduction}
It would not be an exaggeration to say that quantum entanglement has passed the test of time as a fundamental concept in both quantum foundation as well as information. Although it can be considered as `quantum' in nature, its widespread use in information theoretic and communication tasks call for a device-independent characterization. Till date a significant amount of labour has been devoted in this direction. This work humbly intends to contribute to this field, more specifically on the topic of the device independent characterization of a class of permutationally symmetric entangled states.

In the usual tomography task it is possible to estimate the unknown quantum state of system when many identical copies of the same state are available so that different measurements can be performed on each copy. Not only for single system, it is also valid for correlated systems. But in this task the experimenter has to trust the quantumness of the devices  which may not be true in reality. To avoid this difficulty one should be able to device the experiment in such a way that conclusions can be drawn even if the experimenter does not have control over the devices.

Recently preparation and characterization of multipartite states have acquired research interest keeping in view the technological implementation of information theoretic task in a distributive scenario. 14 entangled qubits were prepared via ion-trap experiments very recently. To do full tomography of this multipartite state one needs more information about this state and this procedure has been used for states of low rank, matrix product state and permutationally invariant state. Now this extra knowledge involves the trust the measurement devices. So a different method i.e Device - independent way has been approached in that scenario one does not have detailed knowledge about the experimental apparatus. In this scenario the experimental set up can be realised ad Black box scenario with some input knobs and output data counter. This device independent method is useful for randomness certification, quantum key distribution, dimension witness, certification of entangled measurements , Bit commitment and random number generation. Now this tomography via device independent initially coined as self testing where one multipartite state characterize via only statistical data. Later on large improvisation take place in self testing such as robustness of self testing(both in ideal and in presence of noise).
\par
Among the different types of multiparticle quantum states, Dicke states have attracted a lot of attention. These states were first investigated in 1954 by R. Dicke for describing light emission from a cloud of atoms \cite{dicke54}, and recently several other features have been studied: Dicke states are relatively robust to decoherence\cite{bodoky}, their permutational symmetry allows to simplify the task of state tomography\cite{gezapi,tobipi} and entanglement characterization\cite{gezajosa,gezanjp,huber,novo}. In addition, they are the symmetric states which are in some sense far away from the separable states\cite{dickedefinetti}. Finally, they are relatively easy to generate in photon experiments, and Dicke states with up to six photons have been observed experimentally\cite{hartmut,dicke4,dicke6}.
\par
They propose a method for generating all symmetric Dicke states, either in the long-lived internal levels of N massive particles or in the polarization degrees of freedom of photonic qubits, using linear optical tools only.  They also discuss how our method can also be used to prepare symmetric Dicke states in the polarization degree of freedom of photon qubits. Actually this kind of study motivates to investigate the properties of dicke states. 
\par
In \cite{Navascues13,Navascues14} authors propose the experimental realization of self testing for quantum states using swap method. As an example authors certify singlet fidelity of more than $70\%$ for bipartite case having a CHSH \cite{chsh} violation of $2.57$.In \cite{palNavascues14}, making use of the SWAP method, they extend the concept from the bipartite scenario to the multipartite for the W state \cite{dur}, the $3$ and $4$-qubit GHZ states \cite{ghz} and the 4-qubit cluster state \cite{cluster}. In case of self testing there is no need of any knowledge regarding the specific workings of the experimental devices. In \cite{palNavascues14} they left the question of self testing Dicke states open. We have tried to addressed this question by constructing a permutationally symmetric Bell operator which $n$-qubit Dicke states.
\par
The rest of the article is structured in the following way- in Sec.(\ref{sec1}) we review the concept of Parity invariant Bell inequalities to detect  multipartite entanglement in a Device independent(DI) way. Here we propose and demonstrate a Bell operator that can detect three qubit W state followed by constructing a witness of n qubit Dicke state by providing a permutation invariant Bell operator.

\section{Preliminaries}\label{sec1}
\subsection{PI Bell inequalities}\label{pi}
Bell-type inequalities are the central tool of our investigations\cite{chsh}. We shall focus on multipartite Bell polynomials which are permutationally invariant, that is, they are symmetric under any permutation of the parties. Each observer can choose between two possible measurements featuring binary outputs. We use the following simplified notation to represent such a PI Bell inequality:
\begin{align}\label{sym}
\left[ \alpha_1 \mbox{ } \alpha_2 ; \mbox{ } \alpha_{11} \mbox{ } \alpha_{12} \mbox{ } \alpha_{22} \right]  \equiv&  \alpha_1(A_1 + B_1) + \alpha_2 (A_2 + B_2)  \nonumber  \\
        &+ \alpha_{11}A_1B_1 + \alpha_{12}(A_1B_2 + A_2B_1) \nonumber\\ & +
        \alpha_{22}A_2B_2,
\end{align}
where $A_i=\pm 1$ denotes the outcome of Alice's measurement settings $i=1,2$. Likewise for Bob's settings. The extension to more parties is straightforward. For instance, for $N=3$ parties, the Mermin inequality\cite{mermin}, usually written as
\begin{equation}
M_3 = A_1B_1C_1 - A_1B_2C_2 - A_2B_1C_2 - A_2B_2C_1\leq 2
\end{equation}
now reads
\begin{equation}\label{mermin3} M_3 =  \left[ 0 \mbox{ }0  \mbox{ }; 0\mbox{ } 0\mbox{ } 0\mbox{ } ; 1\mbox{ } 0 \mbox{ }-1\mbox{ }0 \right] \leq 2.
\end{equation}
Here the maximum algebraic sum of $M_3=4$, corresponding to the set of correlations attained with a three-qubit GHZ state~\cite{ghz}:
\begin{equation}
\label{ghzstate3} GHZ_3=(|000\rangle + |111\rangle)/\sqrt 2.
\end{equation}
and Pauli $\hat X$ and $\hat Y$ measurements.

Let us turn to the case of 4 parties. The generalized Mermin-Ardehali-Belinskii-Klyshko\cite{MABK} (MABK) Bell inequality for $N=4$ is given by
\begin{equation}\label{mermin4}
M_4 =  \left[ 0 \mbox{ }0  \mbox{ } ; 0\mbox{ } 0\mbox{ } 0\mbox{
} ; 0\mbox{ } 0 \mbox{ }0\mbox{ }0 ; 1\mbox{ } 1 \mbox{ } -1\mbox{
} -1 \mbox{ } 1 \right] \leq 4.
\end{equation}
Here the quantum maximum reads $8\sqrt 2$, which can be obtained by using $\hat X$ and $\hat Y$ Pauli measurements and a four-qubit GHZ state~\cite{ghz}:
\begin{equation}
\label{ghzstate4} GHZ_4=(|0000\rangle + |1111\rangle)/\sqrt 2.
\end{equation}
\subsection{n-qubit Dicke state}
Here we introduce a class of permutationally symmetric states important from the perspective of quantum information. The $n$-qubit symmetric Dicke state of weight $m(1\leq m<n)$ is defined by,
\be\label{DS}
\ket{m,n}=\frac{1}{\sqrt{n\choose m}}\left[\sum_{j=1}^{n\choose m} \Pi_j(\ket{\overset{m}{\overbrace{111\ldots 11}}\underset{n-m}{\underbrace{000\ldots 00}}})\right],
\ee
where $\{\Pi_j(\ket{1_11_2\ldots 1_m0_{m+1}\ldots 0_n})\}$ is the set of all possible distinct permutations of $m$ 1's and $n-m$ 0's \cite{RP14}. For $m=1$, the state given in (\ref{DS}) is generally called an $n$-qubit W state.
\section{DI witness of multipartite entangled states}
\subsection{Witnessing W state}
In the case of PI Bell inequalities with two binary settings per party, there are nine independent Bell coefficients and we can write the Bell inequality in the notation of section~\ref{pi} as:
\begin{equation}
{\cal B} = \left[ b_1 \mbox{ }b_2  \mbox{ } ; b_3\mbox{ } b_4\mbox{ } b_5\mbox{ } ; b_6\mbox{ } b_7 \mbox{ }b_8\mbox{ }b_9\right] \leq L,
\label{eq:Bell_ineq}
\end{equation}
where $L$ is the local maximum.

Our aim is to construct a Bell operator which gives maximum expectation value for the 3-qubit W state \cite{dur}:
\begin{equation}
|W\rangle\equiv\frac{1}{\sqrt
3}(|001\rangle+|010\rangle+|100\rangle). \label{eq:Wstate}
\end{equation}
The operators of the measurements we have taken are the same for each party, that is $\hat A_1=\hat B_1=\hat C_1\equiv \hat M_1$ and $\hat A_2=\hat B_2=\hat C_2\equiv \hat M_2$. With this choice, observing the permutation symmetry of the W state the Bell operator may be written as
\begin{equation}\label{eq:Bell-operator}
\hat{\cal B}=\sum_{i=1}^6 \hat G_i,
\end{equation}
where
\begin{align}
\hat G_1\equiv& \hat \sigma_z    \hat \sigma_x    \hat \sigma_x\nonumber\\
\hat G_2\equiv& \hat \sigma_x    \hat \sigma_z    \hat \sigma_x\nonumber\\
\hat G_3\equiv& \hat \sigma_x    \hat \sigma_x    \hat \sigma_z\nonumber\\
\hat G_4\equiv& \hat \sigma_z    \hat \sigma_y    \hat \sigma_y\nonumber\\
\hat G_5\equiv& \hat \sigma_y    \hat \sigma_z    \hat \sigma_y\nonumber\\
\hat G_6\equiv& \hat \sigma_y    \hat \sigma_y    \hat \sigma_z \label{eq:B_operators}
\end{align}
Note above we used the shorthand $\hat \sigma_i\hat \sigma_j\hat \sigma_k$ for denoting the tensor product $\hat \sigma_i\otimes \hat \sigma_j\otimes \hat
\sigma_k$. The state giving the maximum quantum violation is the eigenstate belonging to the largest eigenvalue of the Bell-operator with the measurements chosen optimally. Therefore, we must make sure that the W state is an eigenstate of the Bell operator, that is $\langle\psi|\hat{\cal B}|W\rangle=0$ for all states $|\psi\rangle$ orthogonal to
$|W\rangle$. From $\hat \sigma_z|0\rangle=|0\rangle$, $\hat
\sigma_z|1\rangle=-|1\rangle$, $\hat \sigma_x|0\rangle=|1\rangle$, $\hat
\sigma_x|1\rangle=|0\rangle$, $\hat \sigma_y|0\rangle=i|1\rangle$ and $\hat
\sigma_y|1\rangle=-i|0\rangle$ it is not difficult to derive:
\begin{align}
\hat G_1|W\rangle&=|W\rangle-\frac{|100\rangle+|111\rangle}{\sqrt 3}\nonumber\\
\hat G_2|W\rangle&=|W\rangle-\frac{|010\rangle+|111\rangle}{\sqrt 3}\nonumber\\
\hat G_3|W\rangle&=|W\rangle-\frac{|001\rangle+|111\rangle}{\sqrt 3}\nonumber\\
\hat G_4|W\rangle&=|W\rangle-\frac{|100\rangle-|111\rangle}{\sqrt 3}\nonumber\\
\hat G_5|W\rangle&=|W\rangle-\frac{|010\rangle-|111\rangle}{\sqrt 3}\nonumber\\
\hat G_6|W\rangle&=|W\rangle-\frac{|001\rangle-|111\rangle}{\sqrt 3},
\label{eq:HW}
\end{align}
From Eqs.\ (\ref{eq:Bell-operator}) and (\ref{eq:HW}) it follows that $|W\rangle$ is an eigenstate
of $\hat{\cal B}$. The expectation value of $\hat{\cal B}$ is:
\begin{equation}
q\equiv\langle W|\hat{\cal B}|W\rangle=4.
\label{eq:expect}
\end{equation}

Let us stress that the constraints we have derived are only necessary conditions for the W state to be the one which violates the Bell inequality maximally. For the right solution the W state must be the eigenstate belonging to the maximum eigenvalue, and there must not exist another state with some different measurement operators giving the same or larger violation. This extra condition, for instance, is not guaranteed by our procedure.

At this point a pertinent question would be whether one can device a permutation invariant Bell quantity which is maximally violated by a larger class of permutationally symmetric states. We provide such a construction in the next section.

\subsection{Witness of n qubit Dicke states using a PI Bell operator}
Now let us consider a more general scenario where there are n-parties spatially separated from each other. All of them are allowed to perform spin measurements in either X-direction or Y-direction or Z-direction on their respective systems. In such a situation we propose the following construction of a generalized parity invariant Bell operator with an ubiquitous property of witnessing Dicke states:   
\begin{thm}
Let us define a Bell operator for an $n$-qubit system as
\begin{eqnarray}\label{dbell}
\mathbf{B}_n=\sum_{j=1}^{n\choose 2} \Pi_j(\underset{n}{\underbrace{ZZ\ldots ZXZ\ldots ZXZ\ldots Z}})+\nonumber\\\sum_{j=1}^{n\choose 2} \Pi_j(\underset{n}{\underbrace{ZZ\ldots ZYZ\ldots ZYZ\ldots Z}}),
\end{eqnarray}
where $\{\Pi_j(ZZ\ldots Z\underset{Y}{X}Z\ldots Z\underset{Y}{X}Z\ldots Z)\}$ is the set of all possible distinct permutations of two $\underset{Y}{X}$'s and $(n-2)$ Z's and $U\equiv \sigma_U$ for $U\in\{X,Y,Z\}$, then one must have
\begin{equation}
\langle m,n|\mathbf{B}_n|m,n\rangle=\lambda_n
\end{equation}
\end{thm}
\begin{proof}
Here we prove that $|m,n\rangle$ is an eigenstate of the operator $\mathbf{B}_n$. 
\par
Let us define n-qubit operators $$X_{ij}=ZZ\ldots Z\underset{i-th}{X}Z\ldots Z\underset{j-th}{X}Z\ldots Z,$$ $$Y_{ij}=ZZ\ldots Z\underset{i-th}{Y}Z\ldots Z\underset{j-th}{Y}Z\ldots Z,$$ for $1\leq i,j\leq n$ with $i\neq j$.	\\

\be
\begin{array}{rl}
  X_{i,j}|m,n\rangle_{0_i,0_j}&=(-1)^{m}|m+2,n\rangle_{1_i,1_j}\\Y_{i,j}|m,n\rangle_{0_i,0_j}&=(-1)^{m+1}|m+2,n\rangle_{1_i,1_j}\\
   X_{i,j}|m,n\rangle_{1_i,1_j}&=(-1)^{m-2}|m-2,n\rangle_{0_i,0_j}\\Y_{i,j}|m,n\rangle_{1_i,1_j}&=(-1)^{m-1}|m-2,n\rangle_{0_i,0_j}\\
    X_{i,j}|m,n\rangle_{0_i,1_j}&=(-1)^{m-1}|m,n\rangle_{1_i,0_j}\\Y_{i,j}|m,n\rangle_{0_i,1_j}&=(-1)^{m-1}|m,n\rangle_{1_i,0_j}\\
     X_{i,j}|m,n\rangle_{1_i,0_j}&=(-1)^{m-1}|m,n\rangle_{0_i,1_j}\\Y_{i,j}|m,n\rangle_{1_i,0_j}&=(-1)^{m-1}|m,n\rangle_{0_i,1_j}
\end{array}
\ee
Where $|m,n\rangle_{p_i,q_j}$ is any of all possible kets having p $\in \{0,1\}$'s in the i-th position and q $\in \{0,1\}$'s in the j-th position with total $m$ no. of 1's and total $n-m$ no. of 0's.
\\
\\Thus \bean
X_{i,j}|m,n\rangle&=&(-1)^{m-1}|m,n\rangle \\
&&+(-1)^{m-2}\sum\Pi_{i\neq j}\left(|m,n\rangle_{1_i,1_j}\right) \\
&&+(-1)^{m-2}\sum\Pi_{i\neq j}\left(|m,n\rangle_{0_i,0_j}\right)\\
&&+(-1)^{m-2}\sum\Pi_{i\neq j}\left(|m+2,n\rangle_{1_i,1_j}\right)\\
&&+ (-1)^{m-2}\sum\Pi_{i\neq j}\left(|m-2,n\rangle_{0_i,0_j}\right)\\
 Y_{i,j}|m,n\rangle&=&(-1)^{m-1}|m,n\rangle \\
&&+(-1)^{m-2}\sum\Pi_{i\neq j}\left(|m,n\rangle_{1_i,1_j}\right) \\
&&+(-1)^{m-2}\sum\Pi_{i\neq j}\left(|m,n\rangle_{0_i,0_j}\right)\\
&&+(-1)^{m-1}\sum\Pi_{i\neq j}\left(|m+2,n\rangle_{1_i,1_j}\right)\\
&&+ (-1)^{m-1}\sum\Pi_{i\neq j}\left(|m-2,n\rangle_{0_i,0_j}\right)
\eean
Thus \bean
(X_{i,j}+Y_{i,j})|m,n\rangle&=&(-1)^{m-1}2|m,n\rangle \\
&&+(-1)^{m-2}2\sum\Pi_{i\neq j}\left(|m,n\rangle_{1_i,1_j}\right) \\
&&+(-1)^{m-2}2\sum\Pi_{i\neq j}\left(|m,n\rangle_{0_i,0_j}\right) \eean
	Hence \bean
B_n|m,n\rangle &=&(-1)^{m-1}2\sum_{i\neq j}|m,n\rangle\\
&&+(-1)^{m-2}2\sum_{i\neq j}\sum\Pi_{i\neq j}\left(|m,n\rangle_{1_i,1_j}\right) \\
&&+(-1)^{m-2}2\sum_{i\neq j}\sum\Pi_{i\neq j}\left(|m,n\rangle_{0_i,0_j}\right)\\
&=& (-1)^{m-1}2m(n-m)|m,n\rangle
\eean
\end{proof}	
Analytically we have calculated the eigenvalues of the proposed Bell quantity $\mathbf{B}_n$ corresponding to the n-qubit Dicke states as eigenvectors up to $n=10$ (see Appendix \ref{appA}). Interestingly, the eigenvalues are extremal for the cases $m=\lceil{\frac{n}{2}}\rceil, \lfloor{\frac{n}{2}}\rfloor$. One can easily check that the eigenvalues corresponding to $m=\lceil{\frac{n}{2}}\rceil, \lfloor{\frac{n}{2}}\rfloor$ are the maximum among all of the possible eigenstates. Based on the observations presented in Appendix \ref{appA}, we propose the following conjecture- 

	For $m=\lceil{\frac{n}{2}}\rceil$ and $m=\lfloor{\frac{n}{2}}\rfloor$, $|\lambda_n|$ is the largest eigenvalue of $\mathbf{B}_n$ and given by
	\begin{equation}
	|\lambda_n|=2\sum_{j=0}^{n-2} \Mycomb[n-j]{2} (-1)^{j}
	\end{equation}
	with corresponding eigenvector $|m,n\rangle$.

\section{Discussions}
We have derived an analytical witness operator for n-qubit Dicke state in some restricted scenario. Construction of such type of non-local witness operators in general scenario can be further investigated. The work can be extended to other multipartite state. Generalization of this notion to sub-systems having higher dimension can be an interesting line of future research. Our method shall open a way to detect that multipartite dicke state in device independent manner.

\section{Acknowledgements}
 We would like to gratefully acknowledge fruitful discussions with Prof. Guruprasad Kar. We also thank Nirman Ganguly, Bihalan Bhattacharya and Tamal Guha for useful discussions.

\appendix
\section{Analytic results for the extremal eigenvalues and corresponding eigenstates of the Bell operator (Eq.\ref{dbell}) for $n=3$ to $10$}\label{appA}
\begin{widetext}
	\be
	\begin{array}{l|c|c}
		\hline n & \lambda_n & \mbox{ Eigenstate associated with $\lambda_n$ }\\\hline
		3 & -4 & \frac{1}{\sqrt{3}}\left[\ket{001}+\ket{010}+\ket{100}\right]\\
		3 & +4 & \frac{1}{\sqrt{3}}\left[\ket{011}+\ket{101}+\ket{110}\right]\\\hline
		4 & -8 & \frac{1}{\sqrt{6}}\left[\ket{0011}+\ket{0101}+\ket{0110}+\ket{1001}+\ket{1010}+\ket{1100}\right]\\\hline
		5 & -12 & \frac{1}{\sqrt{5\choose 2
			}}\left[\Pi(\ket{00011})\right]\\
			5 & +12 & \frac{1}{\sqrt{5\choose 3
				}}\left[\Pi(\ket{00111})\right]\\\hline
				6 & -18 & \frac{1}{\sqrt{6\choose 3
					}}\left[\Pi(\ket{000111})\right]\\\hline
					7 & -24 & \frac{1}{\sqrt{7\choose 3
						}}\left[\Pi(\ket{0000111})\right]\\
						7 & +24 & \frac{1}{\sqrt{7\choose 4
							}}\left[\Pi(\ket{0001111})\right]\\\hline
							8 & -32 & \frac{1}{\sqrt{8\choose 4
								}}\left[\Pi(\ket{00001111})\right]\\
								\hline
								9 & -40 & \frac{1}{\sqrt{9\choose 4
									}}\left[\Pi(\ket{000001111})\right]\\
									9 & +40 & \frac{1}{\sqrt{9\choose 5
										}}\left[\Pi(\ket{000011111})\right]\\\hline
										10 & -50 & \frac{1}{\sqrt{10\choose 5
											}}\left[\Pi(\ket{0000011111})\right]\\\hline
										\end{array}\label{tab1}
										\ee
									\end{widetext}

\end{document}